\documentclass[11pt]{article} 

\usepackage[utf8]{inputenc} 

\usepackage{geometry} 
\usepackage{graphicx} 
\usepackage{amsmath}
\usepackage{hyperref}
\usepackage{amssymb}
\usepackage{subcaption}

 \DeclareFontFamily{U}{shuffle}{}
 \DeclareFontShape{U}{shuffle}{m}{n}{%
 <5-8>shuffle7%
 <8->shuffle10%
 }{}

 \DeclareSymbolFont{Shuffle}{U}{shuffle}{m}{n}

\DeclareMathSymbol\shuffle{\mathbin}{Shuffle}{"001}
\DeclareMathSymbol\cshuffle{\mathbin}{Shuffle}{"002}

\newtheorem{theorem}{Theorem}[section]

\newtheorem{definition}{Definition}[section]

\newtheorem{example}{Example}[section]

\newtheorem{lemma}{Lemma}[section]

\newtheorem{remark}{Remark}[section]

\newenvironment{proof}[1][Proof]{\noindent\textbf{#1.} }{\ \rule{0.5em}{0.5em}}

\title{A multi-dimensional stream and its signature representation}
\author{Hao Ni }
\begin{document}
\maketitle
\abstract{The signature of a path is an essential object in the theory of rough paths. The signature representation of the data stream can recover standard statistics, e.g. the moments of the data stream. The classification of random walks indicates the advantages of using the signature of a stream as the feature set for machine learning. }
\section{Introduction}
This short paper is devoted to show that the signature of the lead-lag transformation is a useful way to encode a multi-dimensional unstructured data stream. We aim to demonstrate the following points: 
\begin{enumerate}
\item The signature of a discrete sample stream is a rich statistics and encodes the essential information of data stream;
\item The truncated signature of a discrete sample stream provides a summary in terms of the effect of this stream and it leads to dimension reduction for this original stream;
\item The signature of a discrete sample can be used for parameter inference and prediction.
\end{enumerate}
The main result is Theorem \ref{pmoment}, which states that no matter how frequently the path is sampled, the $p^{th}$ moment of the increment process is a linear functional on the truncated signature up to degree $p$.
\section{Notation and Preliminaries}
\subsection{Signatures }
Let us start with introducing the tensor algebra space, in which the signature of a path takes value.
\begin{definition}[Tensor algebra space]
A formal $E$-tensor series is a sequence of tensors $\left( a_{n}\in
E^{\otimes n}\right) _{n\in \mathbb{N}}$ which we write $a=\left(
a_{0},a_{1},\ldots \right) $. There are two binary operations on $E$-tensor series, an addition $+$ and a product $\otimes$, which are defined
as follows. Let $\mathbf{a}=(a_{0},a_{1},...)$ and $\mathbf{b}%
=(b_{0},b_{1},...)$ be two $E$-tensor series. Then we define
\begin{equation}
\mathbf{a}+\mathbf{b}=(a_{0}+b_{0},a_{1}+b_{1},...),
\end{equation}%
and
\begin{equation}
\mathbf{a}\otimes \mathbf{b}=(c_{0},c_{1},...),
\end{equation}%
where for each $n\geq 0$,
\begin{equation}
c_{n}=\sum_{k=0}^{n}a_{k}\otimes b_{n-k}.
\end{equation}%
The product $\mathbf{a}\otimes \mathbf{b}$ is also denoted by $\mathbf{a}%
\mathbf{b}$. We use the notation $\mathbf{1}$ for the series $(1,0,...)$,
and $\mathbf{0}$ for the series $(0,0,...)$. If $\lambda \in \mathbb{R}$,
then we define $\lambda \mathbf{a}$ to be $(\lambda a_{0},\lambda
a_{1},...). $
\end{definition}
\begin{definition}
The space $T\left( \left( E\right) \right) $ is defined to be the vector space of all formal $E$-tensors series.
\end{definition}
Similar to the real valued case, we can define the $\exp$ mapping on $T((E))$ as follows.
\begin{definition}
Let $\mathbf{a}$ be arbitrary element of $T((E))$. Then $\exp(\mathbf{a})$ is the element of $T((E))$ by
\begin{eqnarray*}
\exp(\mathbf{a}):= \sum_{n = 0} \frac{\mathbf{a}^{\otimes n}}{n!}.
\end{eqnarray*}
\end{definition}
Now we are in a position to give the definition of the signature of a path of bounded variation (finite length).
\begin{definition}[Signature of a path]
Let $J$ be a compact interval and $X$ be a continuous function of finite length, which maps $J$ to $E$. The signature $S(X)$ of $X$ over the time interval $J$ is an element $(1,X^{1},...,X^{n},...)$ of $T\left(\left( E\right)\right) $ defined for each $n\geq 1$ as follows
\begin{equation*}
X^{n}=\underset{u_{1}<...<u_{n},\ \ u_{1},...,u_{n}\in J}{\int \cdots \int }%
dX_{u_{1}}\otimes ...\otimes dX_{u_{n}},
\end{equation*}
where the integration is in the sense of Young's integral. The truncated signature of $X$ of order $n$ is denoted by $S^{n}(X)$, i.e. $%
S^{n}(X) = (1,X^{1},...,X^{n})$, for every $n \in \mathbb{N}$.
\end{definition}
\begin{remark}
Suppose that $\{e_{i}\}_{i=1}^{d}$ be a basis of $E$, and thus for every $n \geq 0$, $\{e_{i_{1}} \otimes \dots \otimes e_{i_{n}}\}_{i_{1}, \dots, i_{n}\in\{1, \dots, d\}}$ forms a basis of $E^{\otimes n}$. Therefore $S(X)$ can be rewritten as follows:
\begin{eqnarray*}
S(X) = 1 + \sum_{n=1}^{\infty} \sum_{\underset{\in \{1, \dots, d\}} {i_{1}, \dots, i_{n}}} \left(\underset{\underset{u_{1},...,u_{n}\in J}{u_{1}<...<u_{n}}}{\int \cdots \int
}dX_{u_{1}}^{(i_{1})}dX_{u_{2}}^{(i_{2})} \ldots  dX_{u_{n}}^{(i_{n})}\right) e_{i_{1}} \otimes e_{i_{2}}\cdots \otimes e_{i_{n}}.
\end{eqnarray*}
The signature of a path can be simply regarded as a formal infinite sum of non-commutative tensor products, and the coefficient of each monomial is determined by its corresponding coordinate iterated integral. For every multi-index $I = (i_{1}, \dots, i_{n})$, denote by $\mathbf{X}^{I}$ the following iterated integral of $X$ indexed by $I$, i.e.
\begin{equation*}
\mathbf{X}^{I} = \underset{\underset{u_{1},...,u_{n}\in J}{u_{1}<...<u_{n}}}{\int \cdots \int
}dX_{u_{1}}^{(i_{1})}dX_{u_{2}}^{(i_{2})} \ldots  dX_{u_{n}}^{(i_{n})}.
\end{equation*}
\end{remark}
The first property is Chen's identity (Theorem \ref{ChenIdentity}), which asserts that the signature of the concatenation of two paths is the tensor product of the signature of each path.
\begin{definition}
Let $X:[0,s]\longrightarrow \emph{E}$ and $Y:[s,t]\longrightarrow \emph{E}$
be two continuous paths. Their concatenation is the path $X\ast Y$ defined
by
\begin{equation*}
(X\ast Y)_{u}=\left\{
\begin{array}{ll}
X_{u}, & u\in \left[ 0,s\right]; \\
X_{s}+Y_{u}-Y_{s}, & u\in \left[ s,t\right],%
\end{array}%
\right.
\end{equation*}
where $0 \leq s \leq t$.
\end{definition}
\begin{theorem}[Chen's identity]\label{ChenIdentity}
Let $X:[0,s]\longrightarrow \emph{E}$ and $Y:[s,t]\longrightarrow \emph{E}$
be two continuous paths with finite $1$-variation. Then
\begin{equation}
S(X\ast Y)=S(X)\otimes S(Y),
\end{equation}
where $ 0 \leq s \leq t$.
\end{theorem}
The proof can be found in \cite{RoughPaths}. \newline
Let $\{e_{i}^{*}\}_{i=1}^{d}$ be a basis of the dual space $E^{*}$. Then for every $n \in \mathbb{N}$, $\{e_{i_{1}}^{*} \otimes \cdots \otimes  e_{i_{n}}^{*}\}$ it can be naturally extended to $(E^{*})^{\otimes n}$ by identifying the basis $\left( e_{I} = e_{i_{1}}^{*} \otimes \dots \otimes e_{i_{n}}^{*} \right)$ as
\begin{eqnarray*}
\langle e_{i_{1}}^{*}\otimes \dots \otimes e_{i_{n}}^{*}, e_{j_{1}}\otimes \dots \otimes j_{i_{n}}\rangle = \delta_{i_{1}, j_{1}}\dots \delta_{i_{n}, j_{n}}.
\end{eqnarray*}
The linear action of $(E^{*})^{\otimes n}$ on $E^{\otimes n}$ extends naturally to a linear mapping $(E^{*})^{\otimes n} \rightarrow T((E))^{*}$ defined by
\begin{eqnarray*}
e_{I}(\mathbf{a})= e_{I}^{*}(a_{n}),
\end{eqnarray*}
where $I = (i_{1}, \dots, i_{n})$. \\
Hence the linear forms $e_{I}^{*}$, as $I$ span the set of finite words in the letters $1, \dots, d$ form a basis of $T(E^{*})$. Let $T((E))^{*}$ denote the space of linear forms on $T((E))$ induced by $T(E^{*})$. Let us consider a word $I = (i_{1}, \dots, i_{n})$, where $i_{1}, \dots, i_{n} \in \{1, \dots, d\}$. Define $\pi^{I}$ as $e_{I}^{*}$ restricting the domain to the range of the signatures, denoted by $S(\mathcal{V}^{1}[0, T], E)$, in formula
\begin{eqnarray*}
\pi^{I}(S(X)) = e_{I}^{*}(S(X)),
\end{eqnarray*}
where $X$ is any $E$-valued continuous path of bounded variation.\\
For any two words $I$ and $J$, the pointwise product of two linear forms $\pi^{I}$ and $\pi^{J}$ as real valued functions is a quadratic form on $S(\mathcal{V}^{1}[0, T], E)$, but it is remarkable that it is still a linear form, which is stated in Theorem \ref{shuffle_theorem}. Let us introduce the definition of the shuffle product.
\begin{definition}
We define the set $S_{m,n}$ of $(m, n)$ shuffles to be the subset of permutation in the symmetric group $S_{m+n}$ defined by
\begin{eqnarray*}
S_{m,n} = \{ \sigma \in S_{m+n}: \sigma(1) < \dots < \sigma(m), \sigma(m+1) < \dots < \sigma(m+n)\}.
\end{eqnarray*}
\end{definition}
\begin{definition}
The shuffle product of $\pi^{I}$ and $\pi^{J}$ denoted by $\pi^{I} \shuffle \pi^{J}$ defined as follows:
\begin{eqnarray*}
\pi^{I} \shuffle \pi^{J}  = \sum_{\sigma \in S_{m, n}} \pi^{(k_{\sigma^{-1}(1)}, \dots, k_{\sigma^{-1}(m+n)})},
\end{eqnarray*}
where $I = (i_{1}, i_{2}, \cdots, i_{n}), J =( j_{1},j_{2},\cdots, j_{m})$ and $(k_{1}, \dots, k_{m+n}) =( i_{1}, \cdots, i_{n}, j_{1},$\\
$\cdots, j_{m})$.
\end{definition}
\begin{theorem}[Shuffle Product Property]\label{shuffle_theorem}
Let $X$ be a path of bounded variation. Let $I$ and $J$ be two arbitrary indices. The following identity holds:
\begin{eqnarray*}
\pi^{I}( S(X)) \pi^{J}(S(X)) = (\pi^{I} \shuffle \pi^{J})(S(X)).
\end{eqnarray*}
\end{theorem}
\section{A discrete sampled path and the signature of its lead-lag transformation}
In the following we constrain our discussion on paths observed at a finite number of time stamps and take value in $E:= \mathbb{R}^{d}$.
\subsection{The discrete sampled path and the lead-lag transformation}
Let $\{x_{n}\}_{n = 1}^{L}$ be an increment process, where $x_{n} \in E$. (You can think of it as a return process.) Let $\mathbf{X} : = \{X_{n}\}_{n = 0}^{L}$ denote the corresponding partial sum process of $\{x_{n}\}_{n = 0}^{L-1}$. (It can be thought as a price process.) Mathematically, $\mathbf{X}$ is defined as follows:
\begin{eqnarray*}
X_{0} &=& 0;\\
X_{n+1} &=& \sum_{i  = 1}^{n} x_{i}, \text{ if }n = 1, \dots, L. 
\end{eqnarray*}  
Now let us introduce the lead-lag transformation associated with a $d$-dimensional stream $ \mathbb{X}$ (\cite{flint2013discretely}).
\begin{definition}[Lead-Lag Transformation]
Let $\mathbb{X}:= \{X_{n}\}_{n = 0}^{L}$ be a $d$-dimensional discrete sampled path. The lead-lag transformation associated with $\mathbf{X}$ is a $2d$-dimensional path which is obtained by linear interpolation of $\mathbf{X}: = \{X_{n}\}_{n = 0}^{2L}$,  where $\mathbf{X}^{(i)}_{0} =X_{0}^{(i)}$ and $ \mathbf{X}^{(i)}_{2n-1} = X_{n}^{(i)}$ and for every $n \in \{0, \dots, L-1\}$ and for every $i \in \{1, \dots, d\}$, 
\begin{eqnarray*}
\mathbb{X}^{(i)}_{2n+2} &=& \mathbb{X}^{(i)}_{2n+1} = X_{n+1}^{(i)} \\
\mathbb{X}^{(i+d)}_{2n} &=& \mathbb{X}^{(i+d)}_{2n+1} =X_{n}^{(i)}.
\end{eqnarray*} 
Let $\mathcal{L}$ denote the lead-lag transformation operator. 
\end{definition}
The lead-lag process $\mathbb{X}$ is in the form of the following:\\
\begin{table}
\center
\begin{tabular}{ c c c c c c }
   $\mathbb{X}_{0}$, & $\mathbb{X}_{1}$, &$ \mathbb{X}_{2}$, &...& $\mathbb{X}_{2n-1}$, &$ \mathbb{X}_{2n}$. \\
  $ \vert\vert $ &  $ \vert\vert $& $ \vert\vert $ & & $ \vert\vert $&  $ \vert\vert $\\
$ \left( \begin{array}{c}
X_{0}^{(1)}  \\
X_{0}^{(2)}  \\
\vdots\\
X_{0}^{(d)} \\
X_{0}^{(1)}\\
X_{0}^{(2)} \\
\vdots\\
X_{0}^{(d)} \\ \end{array} \right), $ &  $\left( \begin{array}{c}
X_{1}^{(1)}  \\
X_{1}^{(2)}  \\
\vdots\\
X_{1}^{(d)} \\
X_{0}^{(1)}\\
X_{0}^{(2)} \\
\vdots\\
X_{0}^{(d)} \\ \end{array} \right), $& $ \left( \begin{array}{c}
X_{1}^{(1)}  \\
X_{1}^{(2)}  \\
\vdots\\
X_{1}^{(d)} \\
X_{1}^{(1)}\\
X_{1}^{(2)} \\
\vdots\\
X_{1}^{(d)} \\ \end{array} \right)$ &$\cdots$ & $\left( \begin{array}{c}
X_{n}^{(1)}  \\
X_{n}^{(2)}  \\
\vdots\\
X_{n}^{(d)} \\
X_{n-1}^{(1)}\\
X_{n-1}^{(2)} \\
\vdots\\
X_{n-1}^{(d)} \\ \end{array} \right),  $&  $ \left( \begin{array}{c}
X_{n}^{(1)}  \\
X_{n}^{(2)}  \\
\vdots\\
X_{n}^{(d)} \\
X_{n}^{(1)}\\
X_{n}^{(2)} \\
\vdots\\
X_{n}^{(d)} \\ \end{array}\right) $\\
\end{tabular}
\end{table}
\begin{lemma}[The multiplicative of the lead-lag transformation]
For any two discrete sampled path $\mathbf{X} =  \{X_{n}\}_{n = 0}^{L_{1}}$ and $\mathbf{Y} =  \{Y_{n}\}_{n = 0}^{L_{2}}$
\begin{eqnarray*}
\mathcal{L}(\mathbf{X} * \mathbf{Y} ) = \mathcal{L}(\mathbf{X}) * \mathcal{L}(\mathbf{Y} ),
\end{eqnarray*}
where $\mathbf{X} * \mathbf{Y}$ denote the concatenation of two discrete sampled path, i.e.
\begin{eqnarray*}
(\mathbf{X} * \mathbf{Y})_{n} = \begin{cases} X_{n} &\mbox{if } n \leq L_{1}-1 \\ 
X_{L_{1}} - Y_{0} + Y_{n-L_{1}} & \mbox{if } L_{1} \leq n \leq L_{1} + L_{2}. \end{cases}
\end{eqnarray*}
\end{lemma}
\subsection{The signature of the lead-lag transformation}
Let us define the signature of the discrete sampled stream, and discuss the relevant properties. 
\begin{definition}[The signature representation of a discrete sampled stream] Let $\mathbf{X}$ be a discrete sampled path in $E$ and $\mathbb{X}$ is the lead-lag transformation of $\mathbf{X}$. The signature of $\mathbf{X}$ is defined to be the signature of $\mathbb{X}$, denoted by $S(\mathbb{X})$. Let $S_{d}(\mathbb{X})$ denote the truncated signature of $\mathbb{X}$ up to degree $d$. Let $\mathcal{DS}$ denote the range of signatures of the lead-lag transformation of discrete sampled paths in $E$.
\end{definition}
\begin{lemma}[Chen's Identity for Discrete Sampled Path]
For any two discrete sampled path $\mathbf{X} =  \{X_{n}\}_{n = 0}^{L_{1}}$ and $\mathbf{Y} =  \{Y_{n}\}_{n = 0}^{L_{2}}$.
\begin{eqnarray*}
S(\mathcal{L}(\mathbf{X} * \mathbf{Y} ) ) = S(\mathcal{L}(\mathbf{X}) )\otimes S(\mathcal{L}( \mathbf{Y}) ).
\end{eqnarray*}
\end{lemma}
\begin{definition}[Additive functional on $\mathcal{DS}$] Let $K$ be a linear form on $T((E))$. We say that $K$ is additive in $\mathcal{DS}$ if and only if for every $S(\mathbb{X}), S(\mathbb{Y}) \in \mathcal{DS}$,  it follows that
\begin{eqnarray*}
K(S(\mathbb{X}*\mathbb{Y})) = K(S(\mathbb{X})) +  K(S(\mathbb{Y})). 
\end{eqnarray*}
\end{definition}
For convenience, let us adopt the following notation
\begin{definition}
Fix any positive integer $p$. Let $\mathcal{K}_{I}^{(p)}$ denote the set of the linear forms on $T((E))$ such that it can be written as 
\begin{eqnarray*}
\sum_{\vert J \vert = p, J = (J_{1}, I)} C_{J} \pi^{(J)}
\end{eqnarray*}
where $C_{J} $ are all constants and the summation is taken over all $J$ such that $J$ is of length $p$ and ended in the substring $I$.
\end{definition}
\section{One Dimensional Stream Case}
Let us focus on one dimensional case, and we will show that the signature of $\mathbb{X}$ contains rich information of the path  $\mathbf{X}$ and it is a good basis function to represent the standard statistic, for example, the empirical moments of increments of $\mathbf{X}$ (Theorem \ref{pmoment}). Let us start with discussion on properties of the signature of $\mathbb{X}$. \\
By Chen's identity and simple calculation, the signature of a path in $\mathcal{DS}$ can be given so explicit as follows:
\begin{lemma}[Signature of one-dimensional discrete path]
For any $\mathbb{X} \in \mathcal{DS}$, and $\{x_{i}\}_{i = 1}^{L}$ is the increment process associated with $\mathbb{X} $, then
\begin{eqnarray*}
S(\mathbb{X}) = \bigotimes_{i = 1}^{L} \exp(x_{i} e_{1}) \otimes \exp(x_{i}e_{2})
\end{eqnarray*}
\end{lemma}
\begin{lemma}\label{LemmaSwitchingIndex}
For every index $I$ ending in $2$ and any positive integer $p$,  there exists $K \in \mathcal{K}_{2}^{ (\vert I\vert + p)}$, for any $\mathbb{X}_{L} \in \mathcal{DS}$, such that
\begin{eqnarray*}
\pi^{(I, M_{p})}(S(\mathbb{X}_{L})=  K (S(\mathbb{X}_{L})).
\end{eqnarray*}
where $M_{p}$ is $p$ copies of $1$. \\
For every index $I$ ending in $1$ and any positive integer $p$,  there exists $K \in \mathcal{K}_{1}^{ (\vert I\vert + p)}$, for any $\mathbb{X}_{L} \in \mathcal{DS}$, such that
\begin{eqnarray*}
\pi^{(I, K_{p})}(S(\mathbb{X}_{L})=  K (S(\mathbb{X}_{L})).
\end{eqnarray*}
where $M_{p} = (1, \dots, 1)$, i.e. $p$ copies of $1$. 
\end{lemma}
\begin{proof}
First of all, let us prove that the case $p = 1$. As $I$ ends in $2$, then we can rewrite $I$ as $(J, 2)$. Since $(\pi^{(1)} - \pi^{(2)})(S(\mathbb{X}))= 0$, then
\begin{eqnarray*}
0&=& \pi^{I} (\pi^{(1)} - \pi^{(2)}) = \pi^{(J \shuffle 1,  2)} +\pi^{(I_{2}, 1)} -  \pi^{I_{2}}\shuffle  \pi^{(2)}.\\
\pi^{(I, 1)}& = &   \pi^{I}\shuffle  \pi^{(2)} - \pi^{(J \shuffle 1,  2)} \in \mathcal{K}_{2}^{\vert I \vert +p}.
\end{eqnarray*}
Then we prove this statement by induction on $p$. Let $K_{p}$ be $p$ copies of $2s$.  
\begin{eqnarray*}
0&=& \pi^{I} (\pi^{(M_{p})} - \pi^{(K_{p})}) = \pi^{(J \shuffle M_{p},  2)} +\pi^{(I \shuffle M_{p-1}, 1)} -  \pi^{I}\shuffle  \pi^{K_{p}}.
\end{eqnarray*}
Let us investigate the term $\pi^{(I \shuffle M_{p-1}, 1)}$. 
\begin{eqnarray*}
(I \shuffle M_{p-1}, 1) = (I, M_{p}) + \sum_{k = 1}^{p-1} ( J \shuffle M_{k}, 2, M_{p-k}),
\end{eqnarray*}
and thus
\begin{eqnarray*}
\pi^{(I \shuffle M_{p-1}, 1)}= \pi^{(I, M_{p})} + \sum_{k = 1}^{p-1} \pi^{(J \shuffle M_{k}, 2, M_{p-k})}.
\end{eqnarray*}
For any $k = 1, \dots, p-1$, by induction hypothesis, there exist the linear functional $G \in \mathcal{K}_{2}^{\vert I \vert + p}$ such that for any $S(\mathbb{X}) \in \mathcal{DS}$,
\begin{eqnarray*}
 \pi^{(J \shuffle M_{k}, 2, M_{p-k})}S(\mathbb{X})  = G(S(\mathbb{X})) .
\end{eqnarray*}
Therefore 
\begin{eqnarray*}
\pi^{(I, M_{p})} &=& \pi^{I}\shuffle  \pi^{K_{p}} - \pi^{(J \shuffle M_{p},  2)}  - \sum_{k = 1}^{p-1} \pi^{(J \shuffle M_{k}, 2, M_{p-k})},\\
&=& \pi^{I}\shuffle  \pi^{K_{p}} - \pi^{(J \shuffle M_{p},  2)} - G \in  \mathcal{K}_{2}^{\vert I \vert + p}.
\end{eqnarray*}
Now we complete the first part of the statement. We can use the same strategy to show he second part of the statement.
\end{proof}
\begin{remark}
Since $\pi^{M_{p}} = \pi^{K_{p}}$, Lemma \ref{LemmaSwitchingIndex} shows that for each index $I$, $\pi^{(I)}$ can be rewritten as a linear functional in $\mathcal{K}_{2}^{\vert I \vert}$.
\end{remark}
\begin{lemma}\label{Lemma3}
For any index $I = (i_{1}, \dots, i_{ n -1}, 2)$, and any $S(\mathbb{X}_{L}) \in \mathcal{DS}$,
\begin{eqnarray}
\pi^{(I, 1)}(S(\mathbb{X}_{L})) = \sum_{j = 1}^{L} \pi^{I}(S(\mathbb{X}_{j-1}))x_{j}.\label{Eqn2}
\end{eqnarray}
\end{lemma}
\begin{proof}
We show this lemma by induction on $L$. For $L=1$, both sides of \ref{Eqn2} are equal to $0$. By Chen's identity, for $L \geq 1$, it follows that
\begin{eqnarray*}
&&\pi^{(I, 1)}(S(\mathbb{X}_{L})) = \pi^{(I, 1)}(S(\mathbb{X}_{L-1}) \otimes S(\mathbb{X}_{L-1, L}))  \\
&=& \pi^{(I,1)}(S(\mathbb{X}_{L-1}))  + \pi^{(I)}(S(\mathbb{X}_{L-1}))x_{L}  \\
\end{eqnarray*}
because
\begin{eqnarray*}
S(\mathbb{X}_{L-1, L}) = \exp(x_{L}e_{1}) \otimes \exp(x_{L}e_{2})
\end{eqnarray*}
Then it follows by the induction hypothesis that
\begin{eqnarray*}
\pi^{(I, 1)}(S(\mathbb{X}_{L})) &=& \sum_{j = 1}^{L-1} \pi^{I}(S(\mathbb{X}_{j-1}))x_{j}+\pi^{(I)}(S(\mathbb{X}_{L-1}))x_{L} \\
&=&\sum_{j = 1}^{L} \pi^{I}(S(\mathbb{X}_{j-1}))x_{j}.
\end{eqnarray*}
\end{proof}

\begin{lemma}\label{Lemma4}
For any index $I = (i_{1}, \dots, i_{ n -1}, 2)$ and $k\geq 1$ there exists a linear functional $F$ depending only on $I$ and $k$,  and $F \in \mathcal{K}_{2}^{n+k}$ such that for any $S(\mathbb{X}_{L}) \in \mathcal{DS}$, it holds that
\begin{eqnarray}\label{Eqn1}
F(S(\mathbb{X}_{L})) = \sum_{j = 1}^{L} \pi^{I}(S(\mathbb{X}_{j-1}))x_{j}^{k}.
\end{eqnarray}
\end{lemma}
\begin{proof}
For $k = 1$, it is proved in Lemma \ref{Lemma3}. Assume that $k \leq K-1$ is true. Let us consider the case where $k = K$.
\begin{eqnarray*}
&&\pi^{(I_{2}, 1, \dots, 1)}(S(\mathbb{X}_{L})) - \pi^{(I_{2}, 1, \dots, 1)}(S(\mathbb{X}_{L-1})) \\
& = & \sum_{ j = 1}^{k} \pi^{(I_{2}, 1^{*j})}(S(\mathbb{X}_{L-1})) \frac{x_{L}^{k - j}}{(k-j)!}.
\end{eqnarray*}
After rearranging the above formula we have that
\begin{eqnarray*}
\pi^{(I_{2})}(S(\mathbb{X}_{L-1})) x_{L}^{k} = k!\left( \pi^{(I_{2}, 1, \dots, 1)}(S(\mathbb{X}_{L})) - \pi^{(I_{2}, 1, \dots, 1)}(S(\mathbb{X}_{L-1})) + \sum_{ j = 1}^{k-1} \pi^{(I_{2}, 1^{*j})}(S(\mathbb{X}_{L-1})) \frac{x_{L}^{k-j}}{(k-j)!} \right).
\end{eqnarray*}
By telescope sum of the above equation, we have that
\begin{eqnarray*}
&& \sum_{i = 1}^{L}\pi^{(I_{2})}(S(\mathbb{X}_{i-1})) x_{i}^{k} \\
&=& k!\pi^{(I_{2}, 1, \dots, 1)}(S(\mathbb{X}_{L})) + k! \sum_{i = 1}^{L}\left( \sum_{ j = 1}^{k-1}  \pi^{(I_{2}, 1^{*j})}(S(\mathbb{X}_{i-1})) \frac{x_{i}^{k-j}}{(k-j)!} \right)\\
&=& k!\pi^{(I_{2}, 1, \dots, 1)}(S(\mathbb{X}_{L})) + k! \left( \sum_{ j = 1}^{k-1} \frac{1}{(k-j)!}\sum_{i = 1}^{L} \pi^{(I_{2}, 1^{*j})}(S(\mathbb{X}_{i-1})) x_{i}^{k-j} \right)
\end{eqnarray*}
By Lemma \ref{LemmaSwitchingIndex}, there is a linear functional $G$ depending on $(I_{2}, 1^{*j})$ and $k-j$, such that
\begin{eqnarray*}
\pi^{(I_{2}, 1^{*j})} = G.
\end{eqnarray*}
Then by induction hypothesis,
\begin{eqnarray*}
\sum_{i = 1}^{L} \pi^{(I_{2}, 1^{*j})}(S(\mathbb{X}_{i-1})) x_{i}^{k-j} 
\end{eqnarray*}
can be rewritten as a linear function on $\mathcal{K}_{2}^{n+k}$. $\pi^{(I_{2}, 1, \dots, 1)}$ can be rewritten as a linear functional in $\mathcal{K}_{2}^{n+k}$, so is $ \sum_{i = 1}^{L}\pi^{(I_{2})}(S(\mathbb{X}_{i-1})) x_{i}^{k}$. Now the proof is complete.
\end{proof}

\begin{lemma}\label{Lemma1}
Let $L_{1} \in \mathcal{K}_{1}^{(p)}$ and $L_{1}$ is additive, then there exists $\tilde{L}_{1} \in \mathcal{K}_{1}^{(p+2)} $, such that 
\begin{eqnarray*}
\tilde{L}_{1}(S(\mathbb{X}_{n}))= -\sum_{i = 1}^{n} L_{1}(S(\mathbb{X})_{i})\frac{x_{i}^{2}}{2};
\end{eqnarray*}
\end{lemma}
\begin{proof}
Let $L_{1}:= \sum_{I_{1}} C_{I_{1}} \pi^{(I_{1})}$. For $n \geq 1$, it holds that
\begin{eqnarray*}
&&\pi^{(I_{1},2, 1)}(S(\mathbb{X}_{n}))\\
&=& \pi^{(I_{1},2, 1)}(S(\mathbb{X}_{n-1}) \otimes S(\mathbb{X}_{n-1, n}))\\
&=& \pi^{(I_{1},2, 1)}(S(\mathbb{X}_{n-1}) ) + \pi^{(I_{1},2)}(S(\mathbb{X}_{n-1}))x_{n} .
\end{eqnarray*}
Similarly we have
\begin{eqnarray*}
&&\pi^{(I_{1},2, 2)}(S(\mathbb{X}_{n})) = \pi^{(I_{1},2, 2)}(S(\mathbb{X}_{n-1}) \otimes S(\mathbb{X}_{n-1, n}))\\
& = & \pi^{(I_{1},2, 2)}(S(\mathbb{X}_{n-1}) )+\pi^{(I_{1},2, 2)}(S(\mathbb{X}_{n-1, n}) ) \\
&&+ \pi^{(I_{1},2)}(S(\mathbb{X}_{n-1})) \pi^{(2)}(S(\mathbb{X}_{n-1, n}) ) + \pi^{(I_{1})}(S(\mathbb{X}_{n-1})) \pi^{(2, 2)}(S(\mathbb{X}_{n-1, n}) )+ \mathcal{R}}_{I_{1}(S(\mathbb{X}_{n-1}), x_{n})\\
&=& \pi^{(I_{1},2, 2)}(S(\mathbb{X}_{n-1}) ) +  \pi^{(I_{1})}(S(\mathbb{X}_{n-1, n}) ) \frac{x_{n}^{2}}{2}\\
&&+ \pi^{(I_{1})}(S(\mathbb{X}_{ n-1}))\frac{x_{n}^{2}}{2} +  \pi^{(I_{1},2)}(S(\mathbb{X}_{n-1}))x_{n}\\
&&+ \mathcal{R}_{I_{1}}(S(\mathbb{X}_{n-1}), x_{n}) .
\end{eqnarray*}
where 
\begin{eqnarray*}
\mathcal{R}_{I_{1}}(S(\mathbb{X}_{n-1}), x_{n}) &=& \sum_{J * J_{1} = I_{1}, J \neq \emptyset, J_{1} \neq \emptyset } \pi^{(J)}S(\mathbb{X}_{n-1})\pi^{J_{1} }(S(\mathbb{X}_{n-1, n}))\\
&=&  \sum_{J * J_{1} = I_{1}, J \neq \emptyset, J_{1} \neq \emptyset } \pi^{(J)}S(\mathbb{X}_{n-1})c_{J_{1}} x_{n}^{\vert J_{1}\vert}.
\end{eqnarray*}
The last equality comes from the fact that
\begin{eqnarray*}
\pi^{(I_{1},2, 2)}(S(\mathbb{X}_{n-1, n}) ) &=& \pi^{(I_{1}, 2, 2)}(\exp(x_{n}e_{1}) \otimes \exp(x_{n}e_{2})  )\\
& = & \pi^{(I_{1})}(\exp(x_{n}e_{1})) \pi^{(2,2)}( \exp(x_{n}e_{2})  )\\
& = & \pi^{(I_{1})}(S(\mathbb{X}_{n-1, n}) ) \frac{x_{n}^{2}}{2}.
\end{eqnarray*}
By Lemma \ref{Lemma4}, there exists a linear functional $G_{I_{1}}$ on $\mathcal{K}_{1}^{p+2}$ such that 
\begin{eqnarray*}
G_{I_{1}}(S(\mathbb{X}_{n})) - G_{I_{1}}(S(\mathbb{X}_{n-1})) = \mathcal{R}_{I_{1}}(S(\mathbb{X}_{n-1}), x_{n})  .
\end{eqnarray*}
Thus it follows
\begin{eqnarray*}
&& \pi^{(I_{1},2, 1)}(S(\mathbb{X}_{n})) - \pi^{(I_{1},2, 2)}(S(\mathbb{X}_{n})) \\
&=& \pi^{(I_{1},2, 1)}(S(\mathbb{X}_{n-1}) ) -  \pi^{(I_{1},2, 2)}(S(\mathbb{X}_{n-1}) ) -( \pi^{(I_{1})}(S(\mathbb{X}_{ n-1}))+ \pi^{(I_{1})}(S(\mathbb{X}_{n-1, n}) ) )\frac{x_{n}^{2}}{2} \\
&+& G(S(\mathbb{X}_{n})) - G(S(\mathbb{X}_{n-1})).
\end{eqnarray*}
where 
\begin{eqnarray*}
 G(S(\mathbb{X}_{n})) = \sum_{I_{1}}C_{I_{1}}G_{I_{1}}(S(\mathbb{X}_{n})).
\end{eqnarray*}
Then following the notations
\begin{eqnarray*}
&&\tilde{L}_{1}= \sum_{I_{1}}C_{I_{1}}\left(\pi^{(I_{1},2, 1)}  -\pi^{(I_{1},2, 2)} \right) - G\\
&& f_{n} = \tilde{L}(S(\mathbb{X}_{n}))
\end{eqnarray*}
and it is obviously hat $f(0) = 0$. Moreover since $L_{1}$ is additive, then $L_{1}(S(\mathbb{X}_{ n-1}) ) + L_{1}(S(\mathbb{X}_{ n-1,n})) = L_{1}(S(\mathbb{X}_{ n}) )$, and it follows
\begin{eqnarray*}
f_{n} &=& f_{n-1}  - \sum_{I_{1}}C_{I_{1}}( \pi^{(I_{1})}(S(\mathbb{X}_{ n-1}))+ \pi^{(I_{1})}(S(\mathbb{X}_{n-1, n}) ) )\frac{x_{n}^{2}}{2}\\
&=&f_{n-1} - (L_{1}S(\mathbb{X}_{ n-1})  + L_{1}S(\mathbb{X}_{ n-1,n}))\frac{x_{n}^{2}}{2} \\
&=& f_{n-1} - L_{1}S(\mathbb{X}_{ n}) \frac{x_{n}^{2}}{2} .
\end{eqnarray*}
By the telescoping sum o $f_{n}$, it holds that
\begin{eqnarray*}
f_{n} =  \sum_{i = 1}^{n} (f_{i} - f_{i-1} ) + f_{0}= \sum_{i = 1}^{n}L_{1}S(\mathbb{X}_{ i}) \frac{x_{i}^{2}}{2}. 
\end{eqnarray*}
\end{proof}\\


\begin{theorem}[p-moment]\label{pmoment} For any integer $p>0$, there exist two linear functionals  $L_{p}^{(1)}  \in \mathcal{K}^{(p)}_{1}$,  and $L_{p}^{(2)}  \in \mathcal{K}^{(p)}_{2}$, such that for every path $\mathbb{X}$, the following equation follows:
\begin{eqnarray}\label{p_moment}
L_{p}^{(1)}(S(\mathbb{X})) =  L_{p}^{(2)}(S(\mathbb{X}))= \sum_{i = 1}^{N} x_{i}^{p}.
\end{eqnarray}
Obviously if (\ref{p_moment}) is true, then $L_{p}^{(1)}$ and $L_{p}^{(2)}$ are both additive.
\end{theorem}
\begin{proof}
Let's prove it by induction on $p$. It is true for $p = 1, 2$. Suppose that it holds for $p < P$. Let us study the case when $p = P$.
\begin{eqnarray*}
&&\sum_{i = 1}^{N} x_{i}^{P}  \\
&=& \sum_{i = 1}^{N} \left(L_{p-2}^{(1)}(S(\mathbb{X})_{i}) -  L_{p-2}^{(2)}(S(\mathbb{X})_{i-1})\right) x_{i}^{2} \\
&=&  \sum_{i = 1}^{N} L_{p-2}^{(1)}(S(\mathbb{X})_{i} x_{i}^{2} -  \sum_{i = 1}^{N} L_{p-2}^{(2)}(S(\mathbb{X})_{i-1} x_{i}^{2}
\end{eqnarray*}
By Lemma \ref{Lemma1}, since $L_{p-2}^{(1)}$ is additive, then $\sum_{i = 1}^{N} L_{p-2}^{(1)}(S(\mathbb{X})_{i-1} x_{i}^{2} $ can be rewritten as a linear functional $G \in \mathcal{K}_{p}^{(1)}$ such that 
\begin{eqnarray*}
G_{1}(S(\mathbb{X})_{N}) = \sum_{i = 1}^{N} L_{p-2}^{(1)}(S(\mathbb{X})_{i} x_{i}^{2}.
\end{eqnarray*}
By Lemma \ref{Lemma4}, it follows that there exists $G_{2} \in \mathcal{K}_{p}^{(2)}$, such that
\begin{eqnarray*}
G_{2}(S(\mathbb{X})_{N}) = \sum_{i = 1}^{N} L_{p-2}^{(2)}(S(\mathbb{X})_{i-1} x_{i}^{2}.
\end{eqnarray*}
\end{proof}
\section{ Multi-Dimensional Stream Case}
The following lemma states that the empirical covariance of the increment of a multi-dimensional data stream can be fully characterized by its signatures. 
\begin{lemma}\label{CovarianceLemma}
Let $\mathbf{X}: = \{X_{n}\}_{n = 1}^{L}$ be a $d$-dimensional discretely sampled stream, $\{x_{n}\}_{n = 1}^{L}$ be the associated increment process and $\mathbb{X}$ be the corresponding lead-lag process of $\mathbf{X}$.  For any $i_{1}, i_{2} \in \{1, \dots, d\}$, there exists a linear functional $L$ such that
\begin{eqnarray*}
\sum_{n = 1}^{L} x_{n }^{(i_{1})}x_{n}^{(i_{2})} = 2 \left( \pi^{(i_{1}, i_{2}+d)}(S(\mathbb{X})) -  \pi^{(i_{1}, i_{2})}(S(\mathbb{X})) \right).
\end{eqnarray*}
\end{lemma}
\begin{proof}
For the case that $i_{1} = i_{2 } = i$, it holds that
\begin{eqnarray*}
\sum_{n = 1}^{L} (x_{n }^{(i)})^{2} = \pi^{(i, i+d)}(S(\mathbb{X})) -  \pi^{(i+d, i)}(S(\mathbb{X})) = 2(  \pi^{(i, i+d)}(S(\mathbb{X}))  -  \pi^{(i, i)}(S(\mathbb{X}))).
\end{eqnarray*}
as 
\begin{eqnarray*}
\pi^{(i, i+d)} + \pi^{(i+d, i)}  =  \pi^{(i)} \pi^{(i+d)} =  \pi^{(i)} \pi^{(i)} = \pi^{(i) \shuffle (i)} = 2 \pi^{(i, i)}.
\end{eqnarray*}
For the case that $i_{1} \neq i_{2}$, the signature of the path $\mathbb{X}^{(i_{1}, i_{2})}$, which is the $(i_{1}, i_{2})$ coordinate projection of $\mathbb{X}$ is given as 
\begin{eqnarray*}
S(\mathbb{X}^{(i_{1}, i_{2})}) = \bigotimes_{n=1}^{L} \exp\left( x_{n}^{(i_{1})} e_{i_{1}} +  x_{n}^{(i_{2})} e_{i_{2}} \right)
\end{eqnarray*} 
then it follows that
\begin{eqnarray*}
\pi^{(i_{1}, i_{2})} S(\mathbb{X}) = \sum_{n_{1} < n_{2}}  x_{n_{1}}^{(i_{1})} x_{n_{2}}^{(i_{2})} + \frac{1}{2} \sum_{n = 1}^{L}x_{n}^{(i_{1})} x_{n}^{(i_{2})}.
\end{eqnarray*}
the signature of the path $\mathbb{X}^{(i_{1}, i_{2}+d)}$, which is the $(i_{1}, i_{2}+d)$ coordinate projection of $\mathbb{X}$ is given as 
\begin{eqnarray}\label{Equation1}
S(\mathbb{X}^{(i_{1}, i_{2}+d)}) = \bigotimes_{n=1}^{L} \exp\left( x_{n}^{(i_{1})} e_{i_{1}} \right) \otimes \exp \left( x_{n}^{(i_{2}+d)} e_{i_{2}+d} \right)
\end{eqnarray} 
then it follows that
\begin{eqnarray}\label{Equation2}
\pi^{(i_{1}, i_{2} + d)} S(\mathbb{X}) = \sum_{n_{1} < n_{2}}  x_{n_{1}}^{(i_{1})} x_{n_{2}}^{(i_{2})} +  \sum_{n = 1}^{L}x_{n}^{(i_{1})} x_{n}^{(i_{2})}.
\end{eqnarray}
Combining (\ref{Equation1}) and (\ref{Equation2}), it follows that
\begin{eqnarray*}
 \sum_{n = 1}^{L} x_{n }^{(i_{1})}x_{n}^{(i_{2})}  = 2(\pi^{(i_{1}, i_{2} + d)} (S(\mathbb{X}) )- \pi^{(i_{1}, i_{2})} (S(\mathbb{X})) ).
\end{eqnarray*}
\end{proof}
\begin{lemma}
Let $\mathbf{X}: = \{X_{n}\}_{n = 1}^{L}$ be a $d$-dimensional discretely sampled stream, $\{x_{n}\}_{n = 1}^{L}$ be the associated increment process and $\mathbb{X}$ be the corresponding lead-lag process of $\mathbf{X}$.  For any pairwise different $i_{1}, i_{2}, i_{3} \in \{1, \dots, d\}$, there exists a linear functional $L$ such that
\begin{eqnarray*}
\sum_{n = 1}^{L} x_{n }^{(i_{1})}x_{n}^{(i_{2})}x_{n}^{(i_{3})} = \frac{6}{5} \left(\pi^{(i_{1}, i_{2}, i_{3})}+  \pi^{(i_{1}, i_{2}, i_{3}+d)}) -\pi^{(i_{1}+d, i_{2}, i_{3}+d)}- \pi^{(i_{1}, i_{2}+d, i_{3}+d)}) \right)\left( S(\mathbb{X}) \right).
\end{eqnarray*}
\end{lemma}
\begin{proof}
The signature of the path $\mathbb{X}^{(i_{1}, i_{2}, i_{3})}$, which is the $(i_{1}, i_{2}, i_{3})$ coordinate projection of $\mathbb{X}$ is given as 
\begin{eqnarray*}
S(\mathbb{X}^{(i_{1}, i_{2}, i_{3})}) = \bigotimes_{n=1}^{L} \exp\left( x_{n}^{(i_{1})} e_{i_{1}} +  x_{n}^{(i_{2})} e_{i_{2}} +  x_{n}^{(i_{3})} e_{i_{3}} \right)
\end{eqnarray*} 
then it follows that
\begin{eqnarray*}
\pi^{(i_{1}, i_{2}, i_{3})} S(\mathbb{X}) &=& \sum_{n_{1} <  n_{2} < n_{3}}  x_{n_{1}}^{(i_{1})} x_{n_{2}}^{(i_{2})}x_{n_{3}}^{(i_{3})} + \frac{1}{2}\sum_{n_{1} <  n_{2} }  x_{n_{1}}^{(i_{1})} x_{n_{2}}^{(i_{2})}x_{n_{2}}^{(i_{3})} \\
&&+  \frac{1}{2}\sum_{n_{1}< n_{2} }  x_{n_{1}}^{(i_{1})} x_{n_{1}}^{(i_{2})}x_{n_{2}}^{(i_{3})} +  \frac{1}{6}\sum_{n_{1}=1 }^{L}  x_{n_{1}}^{(i_{1})} x_{n_{1}}^{(i_{2})}x_{n_{1}}^{(i_{3})}.
\end{eqnarray*}
The signature of the path $\mathbb{X}^{(i_{1}, i_{2}+d, i_{3}+d)}$, which is the $(i_{1}, i_{2}, i_{3})$ coordinate projection of $\mathbb{X}$ is given as 
\begin{eqnarray*}
S(\mathbb{X}^{(i_{1}, i_{2}+d, i_{3}+d)}) = \bigotimes_{n=1}^{L} \left( \exp\left( x_{n}^{(i_{1})} e_{i_{1}} \right) \otimes \exp \left(x_{n}^{(i_{2})} e_{i_{2}+d} +  x_{n}^{(i_{3})} e_{i_{3}+d} \right)\right).
\end{eqnarray*} 
then it follows that
\begin{eqnarray*}
\pi^{(i_{1}, i_{2}+d, i_{3}+d)} S(\mathbb{X}) &=& \sum_{n_{1} < n_{2} < n_{3}}  x_{n_{1}}^{(i_{1})} x_{n_{2}}^{(i_{2})}x_{n_{3}}^{(i_{3})}+ \frac{1}{2}\sum_{n_{1} <  n_{2} }  x_{n_{1}}^{(i_{1})} x_{n_{2}}^{(i_{2})}x_{n_{2}}^{(i_{3})} \\
&&+  \sum_{n_{1}< n_{2} }  x_{n_{1}}^{(i_{1})} x_{n_{1}}^{(i_{2})}x_{n_{2}}^{(i_{3})} +  \frac{1}{2}\sum_{n_{1}=1 }^{L}  x_{n_{1}}^{(i_{1})} x_{n_{1}}^{(i_{2})}x_{n_{1}}^{(i_{3})}.
\end{eqnarray*}
Similarly we have that the signature of the path $\mathbb{X}^{(i_{1}, i_{2}, i_{3}+d)}$, which is the $(i_{1}, i_{2}, i_{3}+d)$ coordinate projection of $\mathbb{X}$ is given as 
\begin{eqnarray*}
S(\mathbb{X}^{(i_{1}, i_{2}, i_{3}+d)}) = \bigotimes_{n=1}^{L} \left( \exp\left( x_{n}^{(i_{1})} e_{i_{1}} +  x_{n}^{(i_{2})} e_{i_{2}}  \right) \otimes \exp \left(  x_{n}^{(i_{3})} e_{i_{3}+d} \right)\right).
\end{eqnarray*} 
and thus it holds that
\begin{eqnarray*}
\pi^{(i_{1}, i_{2}, i_{3}+d)} (S(\mathbb{X})) &=& \sum_{n_{1} < n_{2} < n_{3}}  x_{n_{1}}^{(i_{1})} x_{n_{2}}^{(i_{2})}x_{n_{3}}^{(i_{3})}+\sum_{n_{1} <  n_{2} }  x_{n_{1}}^{(i_{1})} x_{n_{2}}^{(i_{2})}x_{n_{2}}^{(i_{3})} \\
&&+ \frac{1}{2} \sum_{n_{1}< n_{2} }  x_{n_{1}}^{(i_{1})} x_{n_{1}}^{(i_{2})}x_{n_{2}}^{(i_{3})} +  \frac{1}{2}\sum_{n_{1}=1 }^{L}  x_{n_{1}}^{(i_{1})} x_{n_{1}}^{(i_{2})}x_{n_{1}}^{(i_{3})}.
\end{eqnarray*}
Moreover we have that
\begin{eqnarray*}
\pi^{(i_{1}+d, i_{2}, i_{3}+d)} (S(\mathbb{X})) &=& \sum_{n_{1} < n_{2} < n_{3}}  x_{n_{1}}^{(i_{1})} x_{n_{2}}^{(i_{2})}x_{n_{3}}^{(i_{3})}+\sum_{n_{1} <  n_{2} }  x_{n_{1}}^{(i_{1})} x_{n_{2}}^{(i_{2})}x_{n_{2}}^{(i_{3})} .
\end{eqnarray*}
Combining the above equations, it follows that
\begin{eqnarray*}
\sum_{n = 1}^{L} x_{n }^{(i_{1})}x_{n}^{(i_{2})}x_{n}^{(i_{3})} = \frac{6}{5} \left(\pi^{(i_{1}, i_{2}, i_{3})}+  \pi^{(i_{1}, i_{2}, i_{3}+d)} -\pi^{(i_{1}+d, i_{2}, i_{3}+d)}- \pi^{(i_{1}, i_{2}+d, i_{3}+d)}\right)\left( S(\mathbb{X}) \right).
\end{eqnarray*}
\end{proof}
\section{ Numerical Examples }
\subsection{Toy Example 1: Correlation estimation}
In this toy example, we want to demonstrate that the signature of a stream can be used as a basis function to represent standard statistics, for example, the mean and the covariance matrix of the increment process. 
\begin{example}
We simulate $400$ samples of the pair $\left\{\rho_{n}, \mathbb{X}_{\rho_{n}}\right\}_{n = 1}^{N=400}$, where $\rho_{n}$ is iid and uniformly distributed in $[0, 1]$, and for each $\rho_{n}$, $\mathbb{X}_{\rho_{n}}$ is generated as a 2-dimensional random walk of length $L$ with the correlation $\rho_{n}$, i.e.
\begin{eqnarray*}
x_{\rho_{n}} \overset{iid}{=} \mathcal{N}\left(0,  \sigma^{2}\left( \begin{array}{cc}1 & \rho_{n}\\
\rho_{n} & 1 \end{array} \right)\right). 
\end{eqnarray*}
How can we estimate the model parameter $\rho$ for each sample path? 
\end{example}
Our method is simply to do the linear regression of the correlation parameter against the truncated signature of the sample path. To better judge the performance of our method, we used the empirical correlation as a benchmark. The empirical correlation for each sample path $\mathbb{X_{\rho}}$ is defined as follows:
\begin{eqnarray*}
\hat{\rho} = \frac{\overset{L-1}{\underset{n = 0}{\sum}} \left(x_{\rho}^{(1)}(n) - \bar{x}_{\rho}^{(1)}\right)\left(x_{\rho}^{(2)}(n) -  \bar{x}_{\rho}^{(2)}\right)}{\sqrt{\overset{L-1}{\underset{n = 0}{\sum}}\left(x_{\rho}^{(1)}(n) - \bar{x}_{\rho}^{(1)}\right)^{2} \overset{L-1}{\underset{n = 0}{\sum}} \left (x_{\rho}^{(2)}(n) -\bar{x}_{\rho}^{(2)}\right)^{2} }}
\end{eqnarray*}
Some parameters I chose are given as follows:
\begin{eqnarray*}
L = 120, N = 200, d = 3
\end{eqnarray*}
\begin{figure}[h!]
\centering
\begin{minipage}[b]{0.45 \linewidth}
     \includegraphics[width=\textwidth]{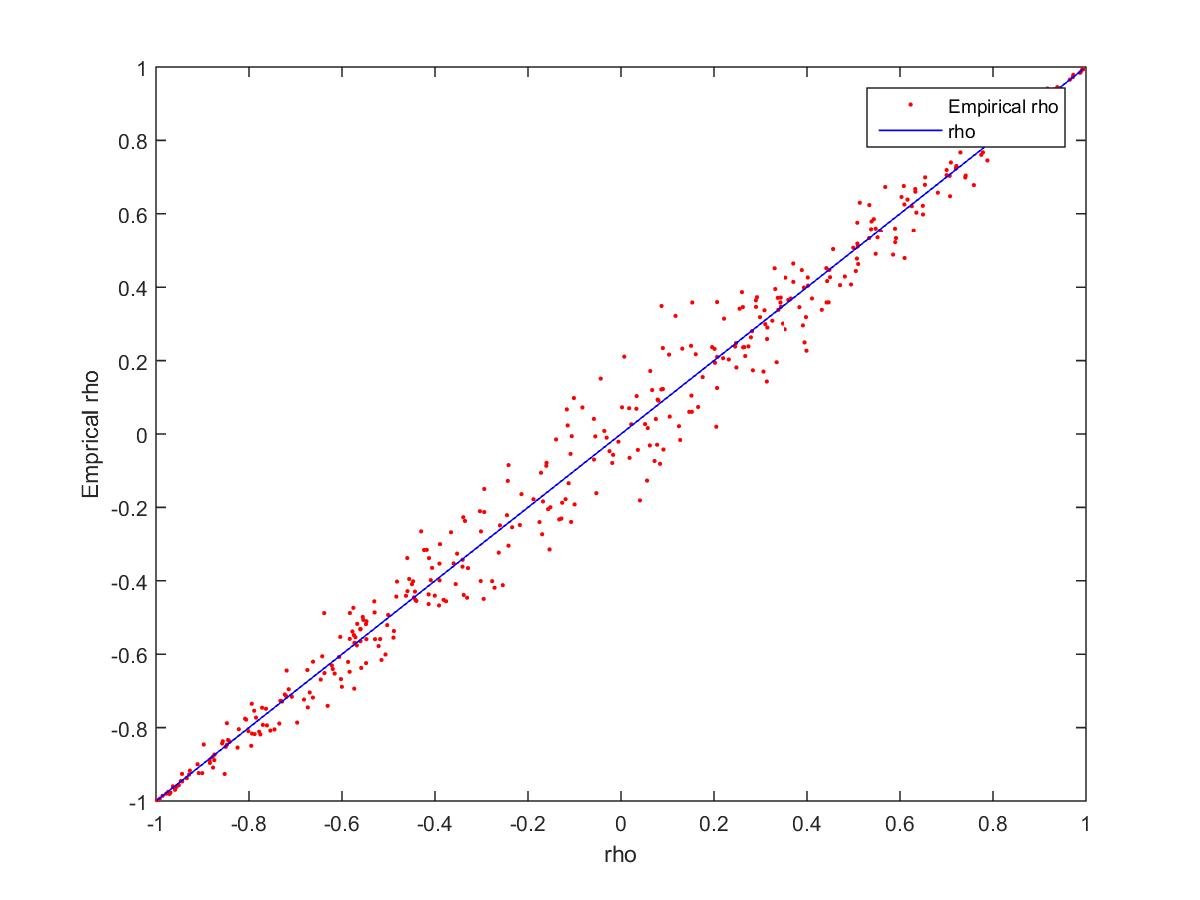} 
\caption{The plot of the empirical correlation v.s the actual correlation}
\end{minipage}
\quad
\begin{minipage}[b]{0.45 \linewidth}
     \includegraphics[width=\textwidth]{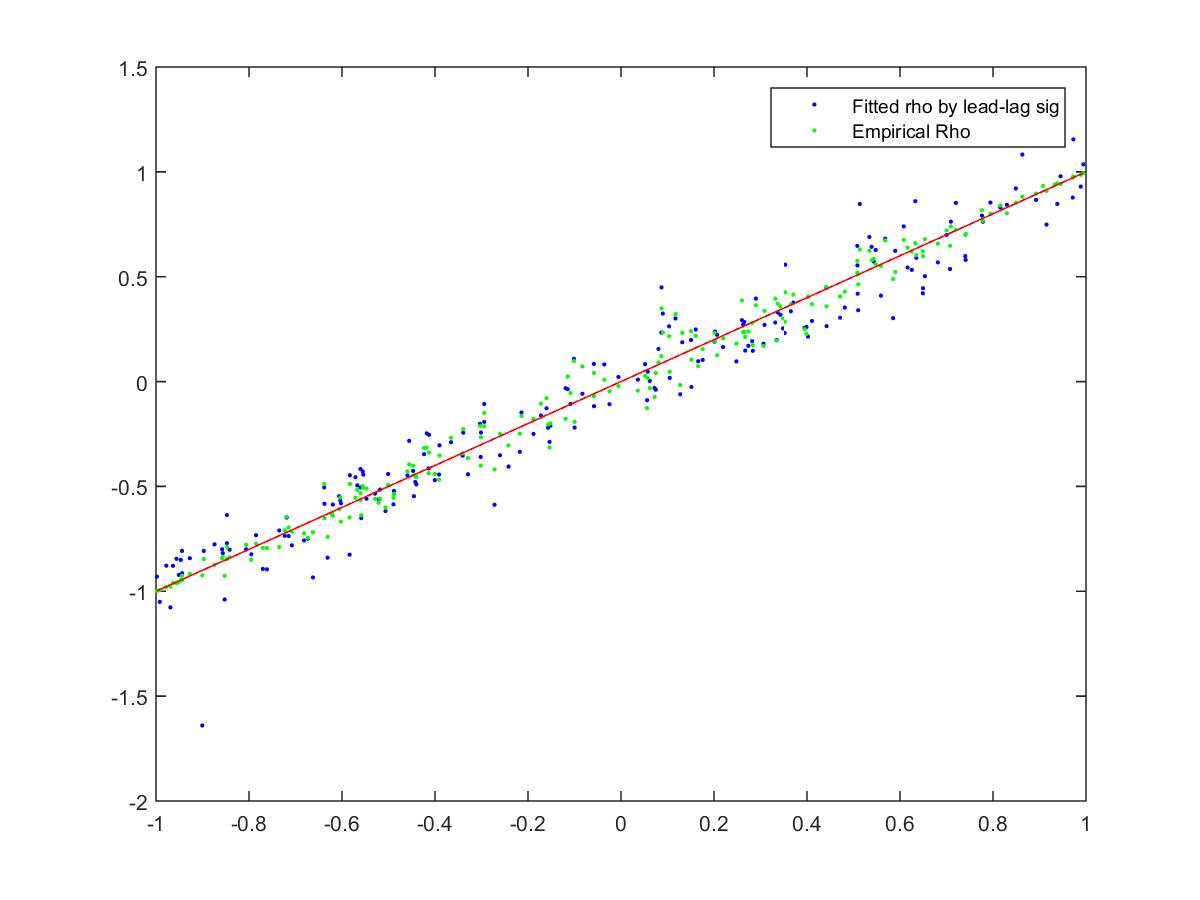}   
\caption{The plot of two estimated correlation against the actual correlation.}\label{SigRho}
\end{minipage}
\end{figure}
Figure \ref{SigRho} shows that the empirical correlation is better in terms of MSE, especially when $\rho$ is near $+1$ an $-1$. However due to the nature of polynomial regression, the signature-approach perform worse when $\rho$ is near the boundary. However the reason why the signature approach is not satisfactory is not because that the truncated signature do not include enough information of the path. Instead the reason is that the regression method we used is too simple and it should be combined with advanced non-linear regression techniques, e.g. rational regression or some local regression methods. Theoretically if properly combined with advanced regression techniques, we should be able to recover the empirical correlation. It is because that by definition of the signature of a stream, Lemma \ref{CovarianceLemma} shows that the empirical covariance/variance of the increment process is a linear combination of the truncated signature up to degree 2, and the ratio of the empirical covariance and the square root of empirical variance of two coordinate increments gives the empirical correlation. 

\subsection{Toy Example 2: Using signatures to classify two classes of random walks}
\begin{example}Let $\mathbb{X}$ denote a standard 3-dimensional random walk of length $L$, and $\mathbb{Y}$ denote the other random walk, where $y^{(1)}, y^{(2)}$ are independent and move to $+1$ and $-1$ with probability $0.5$, but $y^{(3)} = y^{(1)}y^{(2)}$. Given one realization of a random walk of length $L$ generated either by the distribution of $\mathbb{X}$ or that of $\mathbb{Y}$,  which distribution this realized path is from? 
\end{example}
In this example, we can't distinguish which distribution one sample path is generated from by looking at its empirical mean and covariance matrix of the increment distribution, it is simply because that
\begin{eqnarray*}
&&\mathbb{E}[x] = \mathbb{E}[y] = 0;\\
&& \text{cov}[x] = \text{cov}[y] = I_{3}.
\end{eqnarray*} 
But we can almost perfectly classify this sample path using the truncated signatures in this case. We summarize the procedure as follows:
\begin{enumerate}
\item We simulate $N$ paths based on the distribution of $\mathbb{X}$ and $\mathbb{Y}$ respectively. 
\item Compute the truncated signature of those sample paths up to degree $d$. 
\item For each sample path $\mathcal{X}$, let the response variable define in the following way:
\begin{eqnarray*}
f(\mathcal{X}) = \begin{cases} 1 &\mbox{if } \mathcal{X}\text{ is sampled from }\mathbb{X}; \\ 
 0 & \mbox{if }  \mathcal{X} \text{ is sampled from }\mathbb{Y}. \end{cases} 
\end{eqnarray*}
\item We randomly select half of the dataset as the learning set, and the rest data as the backtesting set. Apply SVM classification method to $f(\mathcal{X}) $ against $S(\mathcal{X})_{d}$ in the learning set, where $d=3$.
\item After obtaining the classifier $\hat{f}$, for any new given path $\mathcal{X}^{*}$, by plugging it to the classifier $\hat{f}$, the estimated class of  $\mathcal{X}^{*}$ is given by $\hat{f}(\mathcal{X}^{*})$.
\end{enumerate}
In this example, we choose $N = 200$, $L = 100$ and $d = 3$. The incorrect selection ratio is $1/400$, and it means that there is only one mis-classification for the whole dataset of size $400$. It is noted that the sample space of $\mathbb{Y}$ is actually the subspace of the sample space of $\mathbb{X}$, and theoretically if $\mathcal{X}$ is in the sample space of $\mathbb{X}$, its category is not distinguishable from this sample path trajectory. 

\section{Appendix}

\bibliographystyle{plain}
\bibliography{myref}
\end{document}